\documentclass[journal]{ieeeconf}
\usepackage[cmex10]{amsmath}
\usepackage{amssymb}
\usepackage{stmaryrd}
\usepackage{enumerate}
\usepackage{hyperref}
\hypersetup{
    colorlinks,
    linktocpage,
    linkcolor=blue,
}
\usepackage{tikz}
\usetikzlibrary{external}
\title{Rate Splitting, Superposition Coding and Binning for Groupcasting over the Broadcast Channel: A General Framework}
\author{Henry Romero and Mahesh K. Varanasi
\thanks{This paper was presented in part at the 2017 IEEE Intnl. Symp. on Information Theory (ISIT), Aachen, Germany, 2017 \cite{romero2017rate}.}
\thanks{H. P. Romero was with the Electrical Computer and Energy Engineering Department at the University of Colorado at Boulder when this work was performed and is now with MIT Lincoln Laboratory, 244 Wood St., Lexington, MA, 02420
e-mail: (henry.romero@ll.mit.edu). }
\thanks{M. K. Varanasi is with the Electrical Computer and Energy Engineering Department, 
University of Colorado, Boulder, CO 80303 USA e-mail: (varanasi@colorado.edu).}
}

\begin{document}
\maketitle


\newcommand{\ltwo}[1]{\ell^2(\mathbb{#1})}
\newcommand{\M}[1]{\mathbf{#1}}
\newcommand{\mL}[1]{[1{:}#1]}
\newcommand{\noindex}{\hspace*{-0.8em}}%
\newtheorem{theorem}{Theorem}
\newtheorem{lemma}{Lemma}
\newtheorem{proposition}{Proposition}
\newtheorem{corollary}{Corollary}
\newtheorem{definition}{Definition}
\newcommand{\upset}[1]{\uparrow{#1}}%
\newcommand{\downset}[1]{\downarrow{#1}}%
\newcommand{\labeledop}[2]{\,{\buildrel #1 \over #2 }\,}
\newcommand{\spfont}[1]{\mathsf{#1}} 
\newcommand{\markov}{\mathrel\multimap\joinrel\mathrel-\mspace{-9mu}\joinrel\mathrel-}

\begin{abstract}
A general inner bound is given for the discrete memoryless 
broadcast channel  with an arbitrary number of users  and  general message sets, a setting that accounts for the most general form of concurrent groupcasting, with messages intended for any set of subsets of receivers.  Achievability is based on  superposition 
coding and rate-splitting without and with binning, where each receiver jointly decodes both its desired messages as well as   
the partial interference assigned to it via rate-splitting.
The proof of achievability builds on the techniques for the 
description and analysis of superposition coding recently developed
by the authors for the 
multiple access channel with general messages as well as a new recursive mutual covering lemma for the analysis of the more general achievable scheme with binning.
\end{abstract}

\section{Introduction}
A fundamental feature of wireless transmission is its
broadcast nature. On the one hand, this feature can be 
seen as a detriment  insofar as it inhibits a receiver 
from decoding its desired messages when undesired 
messages interfere at that receiver. On the other hand, 
it facilitates the distribution of one message to many receivers, when all receivers desire that message. 

 A model suited for the analysis of the benefits
 and detractions of this broadcast nature of wireless transmission
 is the broadcast channel (BC) with \emph{general message sets}.
With general message sets, each distinct
message is {\em groupcasted} to a distinct group
of receivers. That group may contain only a single receiver, as is the case for a unicast message,
or as many as all receivers, as is the case for a multicast message,
or any range of intermediate options. 
In general, our model of the broadcast channel with general messages permits multiple such messages to
be concurrently groupcasted.

Much of the work on broadcast channels
has focused on  multiple unicast (i.e., private messages), dating to the seminal
paper on two-user binary symmetric and scalar Gaussian 
 BCs by  Cover \cite{Cover:1972kc}. 
There, a coding strategy known as
\emph{superposition coding} was proposed.
Its extension to general BCs, with an
arbitrary number of users, was developed in
Bergmans \cite{Bergmans:1973ba}. This inner bound
is tight for the  BC with degraded receivers%
\footnote{That is, BCs for which the channel input, followed
by some ordering of the receivers, form a Markov chain.},
 as was established by Gallager \cite{Gallager:1974} in
the  discrete memoryless case, and by
Bergmans \cite{Bergmans:1974uf} in the scalar
Gaussian case. Rate-splitting was first proposed by Carleial in \cite{carleial1978interference} in the context of the two-user interference channel, a technique later used in two- and some three-user broadcast channels (as well as in other problems cf. \cite{el2011network}). Rate-splitting and superposition coding are combined with binning in \cite{liang2007rate} (see \cite[Theorem 8.4]{el2011network}) for the two-user broadcast channel with private and common messages. This idea is extended to the three-user broadcast channel with three degraded messages in \cite{Nair:2009is} (but with one receiver employing indirect joint decoding) and, more generally, to the diamond message set groupcasting in the $K$-user broadcast channel in \cite{salman2020diamond}. The mutual covering lemma (cf. \cite{el2011network}) suffices for the analyses of all these schemes. Moreover, in \cite{Nair:2009is} and \cite{salman2020diamond}, the most economical choice $\spfont{F}=\spfont{E}$ is made.

Here, we considerably generalize these prior inner bounds to the discrete memoryless (DM) BC with an arbitrary number of users and with arbitrary message sets.
Our first inner bound employs generalized notions of superposition coding, rate-splitting, and the joint decoding of  desired messages along with partial decoding of undesired messages. Our second, more general inner bound combines these notions with binning (referred to as multicoding and joint typicality codebook generation in \cite{el2011network}).

To characterize the rates achievable
by superposition coding, we 
use the order-theoretic framework
developed by the authors for the multiple-access channel (MAC) with
general message sets \cite{romero2017unifying}.
 In so doing, we succinctly characterize
 the rates achievable by superposition coding,
and provide a connection to polyhedral combinatorics.



In some cases, it may be beneficial
to decode interference. To allow for
both treating interference as noise
and fully decoding interference, and
a range of intermediate options,
we consider partial interference decoding in a very general way.
This is described by rate-splitting prior to superposition coding, which splits messages into sub-messages (in one of may ways), and relabels each sub-message as being intended by its original intended receivers and by  additional receivers.

Geometrically, this inner bound is equal
to the non-negative rates which lie within
the Minkowski sum of a polytope, representing
the rates achievable through superposition coding 
and a rate-splitting cone of vectors, representing
the enlargements achievable through partial
interference decoding. The polytopes have
 combinatorial structure as they are the intersection of
$K$ unbounded polyhedra, whose bounded
component represents a polymatroid.

Next, our achievable scheme based on rate-splitting and superposition coding is combined with binning. This more general scheme requires for its analysis a generalization of the mutual covering lemma of \cite{el1981proof} (see also \cite[Lemma 8.1]{el2011network}) which we prove here and call the recursive mutual covering lemma. This lemma helps succinctly characterize the conditions under which the probability of encoding errors due to unavailability of jointly typical codewords at the encoder can be made vanishingly small. Here too, we provide a connection to polyhedral combinatorics.

\section{Preliminaries}
\subsection{The Discrete Memoryless Broadcast Channel }
The DM BC consists of 
one transmitter $X\in\mathcal{X}$, $K$ receivers $Y_i\in\mathcal{Y}_i$,
for $1\leq i\leq K$, and a transition function  
$W(y_1,\ldots,y_K|x)$. If $X_t,Y_{1t},\ldots,Y_{Kt}$ are the channel input 
and output at the $t$th channel use, then the conditional probability of 
a sequential block of $n$ channel outputs, conditioned on the 
corresponding $n$ channel inputs, factors as  
$p(y_1^n,\ldots,y_K^n|x^n)= \prod_{t=1}^nW(y_{1t},\ldots,y_{Kt}|x_t)$.
The transmitter may send multiple independent
messages,  each of which is intended for  
a group of receivers.  Each independent message $M_S$, 
and its rate $R_S$, are indexed by the subset $S\subseteq \mL{K}$%
\footnote{For any positive integer $M>0$, 
we denote $\{1,\ldots,M\}$ by $\mL{M}$.}
of the set of receivers that it is intended for.

We collect the indices of all the messages into the set $\spfont{E}$, 
the message index set.  As $\spfont{E}$ contains
non-empty subsets of $\mL{K}$, it is a subset of 
$2^{\mL{K}}$, the power set of $\mL{K}$.
Each receiver only desires a subset of all the messages sent;
we denote the set of indices of messages desired by the $j$th receiver as
\begin{equation}
	\spfont{W}_j^{\spfont{E}} = \{S: j\in S\in\spfont{E}\}  \label{eqn:desired_messages_at_receiver}
\end{equation}

\subsection{Notation}
Let $\spfont{P}$ be an ordered set of sets. If the
order on $\spfont{P}$ satisfies $S \leq S'$ only when $S\subseteq S'$,
so that $S$ and $S'$ are incomparable if neither $S\subseteq S'$
nor $S'\subseteq S$, then we call that order a \emph{superposition order}.
A subset $\spfont{B}$ of $\spfont{P}$ is a \emph{down-set} if, 
for every $S\in \spfont{B}$ and $S'\leq S$, $S'\in \spfont{B}$.
Similarly, $\spfont{B}$ is an up-set if, for every $S\in \spfont{B}$
 and $S\leq S'$, $S'\in \spfont{B}$. Let $\downset{\spfont{Q}}$ 
 (resp., $\upset{\spfont{Q}}$) 
 be the smallest down-set (resp., up-set) of $\spfont{P}$ 
 containing $\spfont{Q}$ \cite{DaveyPriestley:2002}.
 Let $\mathcal{F}_{\downarrow}(\spfont{P})$ 
 be the set of all down-sets of $\spfont{P}$, which is closed under 
 intersections and unions, and is referred to as the 
 \emph{down-set lattice}.%
\footnote{The set of down-sets is a lattice whose join and meet
are given by  union and intersection, respectively, as
the set of downs-set is closed under union and intersection.}
If $\mathcal{I}$ contains the indices for a collection of random 
variables $(A_i:i\in\mathcal{I})$, we succinctly denote
any subcollection $(A_i:i\in\mathcal{J})$, with
$\mathcal{J}\subseteq\mathcal{I}$, of those random variables 
as $A_{\mathcal{J}}$. 

In the context of examples, albeit with an abuse of notation, 
we denote the index of a message or its rate by a string of 
the elements of that set. For example, the 
message $M_{\{1,2,4\}}$ and its rate $R_{\{1,2,4\}}$ are 
denoted as $M_{124}$ and $R_{124}$, respectively. 
Similarly, the message index set 
$\spfont{E}=\{\{1\}, \{2\}, \{1,2\}, \{1,2,3\}\}$ will be 
succinctly denoted by $\spfont{E} = \{1,2,12,123\}$.

\section{Superposition-Based Inner Bound}

Consider an arbitrary $K$-receiver DM BC,
whose channel input takes values in $\mathcal{X}$
and whose message index set is $\spfont{E}$.
For any superset $\spfont{F}$ of $\spfont{E}$
within $\spfont{P} = 2^{\mL{K}}$, let   $X' = (X,U_{\spfont{F}}, Q)$
refer to a collection of random variables which
includes the channel input $X$, 
$|\spfont{F}|$ auxiliary random variables 
$U_{\spfont{F}} = (U_S:S\in\spfont{F})$,
and a coded time-sharing random variable $Q$.
Let $\spfont{F}$ be ordered by a
superposition order (denoted as $\leq$).

\begin{definition}\label{def:superposition_admissibleRVs}
$X' = (X,U_{\spfont{F}}, Q)$  is
\emph{superposition-admitting}  if 
$X$ is a deterministic function of $(U_{\spfont{F}},Q)$,
whose joint probability mass function factors as
\begin{equation}	
	p(U_{\spfont{F}},Q) = p(Q)\prod\nolimits_{S\in\spfont{F}}p(U_S|U_{\upset{\{S\}}\backslash \{S\}},Q) 	
	\label{eqn:pmf_factorization}
 \end{equation}
Let $\mathcal{A}_{(\leq)}^{\spfont{F}}$ contain all 
superposition-admitting random variables with respect to the ordered
set  $\spfont{F}$.
\end{definition}

For any set of random variables within $X'\in \mathcal{A}_{(\leq)}^{\spfont{F}}$,
define%
\footnote{For a set $E$, we use the notation $\mathbb{R}^E$ to refer 
to a vector of real numbers, where the components are labeled
 by the elements of $E$. Thus, if $E$ has $M$ elements, this
 space can be identified with $\mathbb{R}^M$. Replacing
 $\mathbb{R}$ with  $\mathbb{R}_+$  denotes the positive orthant
 of the previously described spaces.} 
\begin{IEEEeqnarray}{l}
	\mathcal{P}_{\downarrow}^{(j)}(X';\spfont{F}) =  \big\{	R \in \mathbb{R}_+^{\spfont{F}}: \IEEEnonumber \\
		\quad \sum_{S\in \spfont{B}} R_{S} \leq I(U_\spfont{B};Y_j|U_{\spfont{W}_j^{\spfont{F}}\backslash \spfont{B}},Q) 
				\text{ for all } \spfont{B}\in \mathcal{F}_{\downarrow}(\spfont{W}_j^{\spfont{F}})	\big\}
				 \IEEEeqnarraynumspace\label{eqn:downsetpolymatroid}
\end{IEEEeqnarray}
a subset of $\mathbb{R}_+^{\spfont{F}}$. Note that this polyhedron 
imposes no constraints on the nonnegative rates
$(R_S:S\in\spfont{F}\backslash \spfont{W}_j^{\spfont{F}})$.
Our principal result is the following inner bound.
\begin{theorem}\label{thm:superposition_DMBC_IB}
For the $K$-receiver DM BC with general message sets, 
the non-negative rates $(R_S:S\in \spfont{E})$ are achievable if,
for a message index superset $\spfont{F}$ with $\spfont{E}\subseteq \spfont{F}$
equipped with a superposition order, there
 exist non-negative split-rates $(r_{S\to S'}: S\in\spfont{E},S'\in \spfont{F}, S\subseteq S')$,
where the desired rates satisfy
\begin{equation} R_S = \sum\nolimits_{S\subseteq S': S'\in \spfont{F}} r_{S\to S'}	\quad \text{ for all }S\in \spfont{E}
		\label{eqn:targetrates_BC}\end{equation}
while the reconstructed rates 
\begin{equation} 	
	\hat{R}_{S'} = \sum\nolimits_{S\subseteq S': S\in \spfont{E}}r_{S\to S'}	\quad \text{ for all }S'\in \spfont{F} 
	\label{eqn:reconstructedrates_BC}\end{equation}
are constrained to be within the rate region
\begin{IEEEeqnarray}{c}
		\bigcup_{X'\in \mathcal{A}_{(\leq)}^{\spfont{F}}} 
			\left(\mathcal{P}_{\downarrow}^{(1)}(X';\spfont{F})\cap
		 \cdots \cap \mathcal{P}_{\downarrow}^{(K)}(X';\spfont{F})\right)  \label{eqn:packing_constraint_BC} \\
		\IEEEnonumber
\end{IEEEeqnarray}
Furthermore, the projection of each polyhedron $ \mathcal{P}_{\downarrow}^{(j)}(X';\spfont{F}) $ onto the cone of non-negative rates with indices in $\spfont{W}_j^{\spfont{F}}$ is a polymatroid.
\end{theorem}
\begin{proof}
We sketch the achievability proof in two parts.  First, we show that 
a set of rates $(\hat{R}_{S'}:S'\in\spfont{F})$ are achievable for
the enlarged message index set $\spfont{F}$ if they
are within the rate region \eqref{eqn:packing_constraint_BC}.
This follows from superposition encoding 
at the single transmitter, where we use the order-theoretic framework we developed in \cite{romero2017unifying}  to analyze the 
probability of error at each receiver.
Details for this argument are provided in the Appendix. 
 
 Next,  we allow for partially decoding interference
  through rate-splitting. Divide each message $M_S$
into a collection of split-messages
 $(m_{S\to S'}:S\in\spfont{E},S'\in\spfont{F}, S\subseteq S')$.
 In this collection, the partial message $m_{S\to S'}$ 
 is to be treated as though it were intended for all
  messages within the set $S'$
rather than only for the messages within the set $S$.
Whenever the receiver decodes a partial message $m_{S\to S'}$,
with $S$ as a strict subset of $S'$, that receiver is partially decodes the interfering
message $m_S$.
With  $r_{S\to S'}$ as the rate of the split message $m_{S\to S'}$,
this rate-split decomposes
the target rate $R_S$ according to \eqref{eqn:targetrates_BC}.
Each receiver within the group $S' \in\spfont{F}$ must decode the
reconstructed message 
  $\hat{M}_{S'} = (m_{S\to S'}: S\in\spfont{E} : S\subseteq S')$,
  whose rate $\hat{R}_S$ is given by \eqref{eqn:reconstructedrates_BC}. 
  These reconstructed rates 
in turn can be reliably transmitted to their desired receivers if they satisfy
\eqref{eqn:packing_constraint_BC}, as previously mentioned. 

Finally, the fact that the projection of $ \mathcal{P}_{\downarrow}^{(j)}(X';\spfont{F}) $
 onto the cone of non-negative rates with indices in $\spfont{W}_j^{\spfont{F}}$ is a polymatroid
follows from Theorem 3 of \cite{romero2017unifying}. 
\end{proof}

\subsection{Inner bound of Theorem \ref{thm:superposition_DMBC_IB} as a Minkowski sum}
To delineate the two strategies that comprise 
Theorem \ref{thm:superposition_DMBC_IB}, and to better
understand the geometric structure of the associated
inner bound, we express that inner bound
as the Minkowski%
\footnote{For any two subsets $\mathcal{A}$ and $\mathcal{B}$ of 
$\mathbb{R}^{\spfont{F}}$,  
$\mathcal{A} + \mathcal{B} = \{a + b : a\in \mathcal{A}, b\in \mathcal{B}\}$ 
is the Minkowski sum of $\mathcal{A}$ and $\mathcal{B}$. }
 sum of a polytope, representing the rates achievable
 through superposition coding, and a cone,
representing the rate gains possible through
partial interference decoding.

Observe that if $R_{\spfont{E}}$ is in the inner bound  of Theorem 
\ref{thm:superposition_DMBC_IB}, then for each $S\in\spfont{E}$,
\begin{IEEEeqnarray*}{rl}
	R_S &= r_{S\to S} + \sum\nolimits_{S'\in\spfont{F}:S'\supset S} r_{S\to S'} \\
		&= \left(\hat{R}_S - \sum\nolimits_{S'\in\spfont{E}: S'\subset S}r_{S'\to S}\right) 
			+ \sum\nolimits_{S'\in\spfont{F}:S'\supset S} r_{S\to S'}  \\
		&= \hat{R}_S  + \Delta_S \IEEEyesnumber \label{eqn:ratedecomposition}
\end{IEEEeqnarray*}
where 
\begin{equation}
	\Delta_S = \sum\nolimits_{S'\in\spfont{F}: S'\supset S} r_{S\to S'} 
	- \sum\nolimits_{S'\in\spfont{E}: S'\subset S} r_{S'\to S}	\label{eqn:exchangerates1}
\end{equation}
These equations reveal that the rate-splitting  effectively 
exchanges groupcasting rates between different groupcast labels.
 As some of these exchanges occur not just
 among the rates indexed by the message index set $\spfont{E}$,
 but among the rates within the enlarged message index
 set $\spfont{F}$, \eqref{eqn:ratedecomposition} is incomplete.
 A complete description embeds the original set of rate demands, 
 which live in $\mathbb{R}_+^\spfont{E}$, into $\mathbb{R}_+^F$,
 by setting $R_S = 0$ for each $S\in\spfont{F}\backslash \spfont{E}$. 
 For each $S\in \spfont{F}$, the analogous statement to
 \eqref{eqn:ratedecomposition} is that $R_S = \hat{R}_S  + \Delta_S $, where
\begin{equation}
	\Delta_S = - \sum\nolimits_{S'\in\spfont{E}: S'\subset S} r_{S'\to S}\label{eqn:exchangerates2}
\end{equation}

To translate these observations into a geometrical
characterization of the inner bound in Theorem  \ref{thm:superposition_DMBC_IB},
we introduce a few additional definitions. For each
pair of distinct message indices $S\in\spfont{E},S'\in\spfont{F}$ with 
$S\subset S'$, let the vector $e_{S\to S'}$ be the vector in 
$\mathbb{R}_+^{\spfont{F}}$ for which
\begin{equation}	(e_{S\to S'})_A = \begin{cases}
				1 & \text{ if }A= S \\
				-1 & \text{ if }A = S'\\ 
				0 & \text{ else}	\end{cases}	\label{eqn:generating_vectors}	\end{equation}
Let $\mathcal{C}_{\uparrow}^{\spfont{F}}$ denote the cone of vectors  
generated by the collection of vectors 
$\{e_{S\to S'} : S\in\spfont{E},S'\in\spfont{F}\text{ and } S\subset S' \}$.
Then the vector $\Delta = (\Delta_S:S\in\spfont{F})$, whose elements
are defined by \eqref{eqn:exchangerates1}, if $S\in\spfont{E}$,
or by \eqref{eqn:exchangerates2}, if $S\in\spfont{F}\backslash \spfont{E}$,
is within the cone  $\mathcal{C}_{\uparrow}^{\spfont{F}}$. 
We will refer to this vector as the exchange-rate vector, 
as its elements reveal the effective rate exchange that rate-splitting entails.
Recalling that the reconstructed rate point $\hat{R} = (\hat{R}_S:S\in\spfont{F})$ 
is within the union of polytopes described by \eqref{eqn:packing_constraint_BC},
 these observations lead to the following geometrical characterization
 of Theorem  \ref{thm:superposition_DMBC_IB}.

\begin{theorem}\label{thm:innerbound_exchangeratedescription}
For the $K$-receiver broadcast channel with message index set $\spfont{E}$, 
if $\spfont{F}$ is a superset of $\spfont{E}$ and is equipped with a superposition order, 
then the non-negative rates $(R_S:S\in\spfont{E})$ within
\begin{equation}	
		\bigcup_{X'\in\mathcal{A}_{(\leq)}^{\spfont{F}}}	 
		\left( \bigcap_{j=1}^K\mathcal{P}_{\downarrow}^{(j)}(X';\spfont{F})  
				+ \mathcal{C}_{\uparrow}^{\spfont{F}}\right)\cap \mathbb{R}_+^{\spfont{F}\to \spfont{E}} 
	 	\label{eqn:direct_sum_IB_BC}
\end{equation}
are achievable, where 
\begin{IEEEeqnarray}{c}	\mathbb{R}_+^{\spfont{F}\to \spfont{E}} = \{	R\in\mathbb{R}_+^\spfont{F}: 
	R_S = 0\text{ if } S\in \spfont{F}\backslash \spfont{E}	\}  \label{eqn:projection_operator} \\ 		\IEEEnonumber
\end{IEEEeqnarray}
\end{theorem}

\subsection{Specializations \label{sec:nosuperposition}}
We highlight two specializations of Theorem \ref{thm:superposition_DMBC_IB},
 corresponding to two  alternative choices of superposition order on
 the superset $\spfont{F}$ of $\spfont{E}$. In general, the choice
 of  superposition order corresponds to a choice in dependency in the generation of auxiliary codewords.

One possibility is to generate auxiliary codewords that are conditionally
dependent whenever possible. This corresponds to 
equipping $\spfont{F}$ with the superposition order of set inclusion. 
In this case, the set of superposition-admitting random
variables, $\mathcal{A}_{(\subseteq)}^{\spfont{F}}$, contains
those random variables $X' = (X,U_{\spfont{F}},Q)$
where the density of the auxiliary random variables and time-sharing
random variable factors
  as $p(Q)\prod_{S\in\spfont{F}}p(U_S|(U_{S'}: S\subset S'),Q)$.
 The constituent polyhedra are denoted by
$\mathcal{P}_{(\subseteq)}^{(j)}(X';\spfont{F})$,
as defined in \eqref{eqn:downsetpolymatroid}, where
a subset $\spfont{B}$ of $\spfont{W}_j^{\spfont{F}}$ is within
$\mathcal{F}_{(\subseteq)}(\spfont{W}_j^{\spfont{F}})$ if, for every $S\in\spfont{B}$,
and $S'\subseteq S$ with $S'\in \spfont{F}$, $S'\in \spfont{B}$.
With this superposition order, Theorem \ref{thm:superposition_DMBC_IB} yields
\begin{corollary}\label{corr:fullsuperposition_DMBC_IB}
For the $K$-receiver broadcast channel with message index set $\spfont{E}$, 
if $\spfont{F}$ is a superset of $\spfont{E}$ ordered by set inclusion, then
the non-negative rates $(R_S:S\in \spfont{E})$ within the rate region
\begin{equation}	
		\bigcup_{X'\in\mathcal{A}_{(\subseteq)}^{\spfont{F}}}	
			\left( \bigcap_{j=1}^K\mathcal{P}_{(\subseteq)}^{(j)}(X';\spfont{F}) 
				+ \mathcal{C}_{\uparrow}^{\spfont{F}}\right)\cap \mathbb{R}_+^{\spfont{F}\to \spfont{E}} 
\end{equation}
where  $X'\in \mathcal{A}_{(\subseteq)}^{\spfont{F}}$, are achievable.
\end{corollary}

Another possibility is to generate auxiliary codewords that are all
conditionally independent given a coded time-sharing sequence. This corresponds to 
equipping $\spfont{F}$ with the discrete order, so that 
every pair of sets within $\spfont{F}$ are incomparable. 
The set of superposition admitting random variables,
$\mathcal{A}_{(=)}^{\spfont{F}}$, contains those random variables
$X' = (X,U_{\spfont{F}},Q)$ for which 
the auxiliary random variables $(U_S:S\in \spfont{F})$
are mutually independent, conditioned on the time-sharing random variable $Q$. 
 The constituent polyhedra are denoted by
$\mathcal{P}_{(=)}^{(j)}(X';\spfont{F})$, where the
down-set lattice $\mathcal{F}_{(\subseteq)}(\spfont{W}_j^{\spfont{F}})$
contains every subset $\spfont{B}\subseteq \spfont{W}_j^{\spfont{F}}$.
With this superposition order, Theorem \ref{thm:superposition_DMBC_IB} yields
\begin{corollary}\label{corr:nosuperposition_DMBC_IB}
For the $K$-receiver broadcast channel with message index set $\spfont{E}$, 
if $\spfont{F}$ is a superset of $\spfont{E}$ ordered by the discrete order, then
the non-negative rates $(R_S:S\in \spfont{E})$ within the rate region
\begin{equation}	
		\bigcup_{X'\in\mathcal{A}_{(=)}^{\spfont{F}}}	
		\left( \bigcap_{j=1}^K\mathcal{P}_{(=)}^{(j)}(X';\spfont{F}) 
				+ \mathcal{C}_{\uparrow}^{\spfont{F}}\right)\cap \mathbb{R}_+^{\spfont{F}\to \spfont{E}} 
\end{equation}
where $X'\in \mathcal{A}_{(=)}^{\spfont{F}}$, are achievable.
\end{corollary}

\section{Examples of rate regions for two- and three-receiver BCs}
In this section, we show that the inner bound in Theorem
\ref{thm:superposition_DMBC_IB} 
is sufficiently general to include previous characterizations of
superposition coding in the BC.

\subsection{Two-receiver BC}
\subsubsection{Degraded Messages}
Consider the two-user BC in 
the special case where $R_2 = 0$. This is
the degraded message case, with $\spfont{E} = \{1,12\}$. 
Korner and Marton \cite{Korner:1977tl} determine
that the capacity region contains the non-negative  rates for which
\begin{IEEEeqnarray}{rl}	\IEEEyesnumber
	R_{12} &\leq I(U_{12};Y_2)\IEEEyessubnumber	 \\
	R_1 &\leq  I(U_1;Y_1|U_{12}) \IEEEyessubnumber \\
	R_1 + R_{12} &\leq  I(U_1;Y_1) \IEEEyessubnumber
\end{IEEEeqnarray}
for some pair $(U_1,U_{12})$, where $X=U_1$
and $U_{12}\markov U_1 \markov Y$ form a Markov chain.
Achievability of this rate region follows from Corollary
\ref{corr:fullsuperposition_DMBC_IB}  when $\spfont{F} = \spfont{E} = \{1,12\}$
and both rate-splitting and time-sharing are omitted.%

\subsubsection{Cover's Inner Bound}
For the two-receiver DM BC with message index set
 $\spfont{E} = \{1,2,12\}$,  Cover \cite{Cover:1975kj}
determines that any non-negative rate point $(R_1,R_2,R_{12})$
is achievable if 
\begin{IEEEeqnarray}{rl}\IEEEyesnumber
	R_1 &\leq I(U_1;U_{12},Y_1|Q) \IEEEyessubnumber\label{eqn:cover_polymatroid1a}\\
	R_{12} &\leq I(U_{12};U_1,Y_1|Q) \IEEEyessubnumber \\
	R_1 + R_{12} &\leq I(U_1,U_{12};Y_1|Q) \IEEEyessubnumber \label{eqn:cover_polymatroid1c}\\[5pt]
	R_2 &\leq I(U_2;U_{12},Y_2|Q)  \IEEEyessubnumber   \label{eqn:cover_polymatroid2a}  \\
	R_{12} &\leq I(U_{12};U_2,Y_2|Q) \IEEEyessubnumber\\
	R_2 + R_{12} &\leq I(U_2,U_{12};Y_2|Q) \IEEEyessubnumber \label{eqn:cover_polymatroid2c}\ 
\end{IEEEeqnarray}
for a channel input $X$ that is a deterministic function of a 
time-sharing random variable $Q$  and a set
of auxiliary random variables $U_1,U_2,U_{12}$  that are independent  
when conditioned on $Q$. 

This inner bound follows from Corollary \ref{corr:nosuperposition_DMBC_IB} 
when $\spfont{F} = \spfont{E}$ and rate-splitting is omitted.
By the conditional independence
of the auxiliary random variables,
$I(U_{\spfont{B}};U_{\{1,12\}\backslash \spfont{B}},Y_1|Q)
		= I(U_{\spfont{B}};Y_1|U_{\{1,12\}\backslash \spfont{B}},Q)$
for each non-empty subset $\spfont{B}\subseteq \{1,2\}$,
and analogously for the terms involving the second receiver.
Hence, the conditions that 
$(R_1,R_2,R_{12})$ be within the polyhedron $\mathcal{P}_{(=)}^{(1)}(X';\{1,2,12\})$
and $\mathcal{P}_{(=)}^{(2)}(X';\{1,2,12\})$, when $X' = (X,U_1,U_2,U_{12},Q)$,
are the conditions \eqref{eqn:cover_polymatroid1a}-\eqref{eqn:cover_polymatroid1c}
and   \eqref{eqn:cover_polymatroid2a}-\eqref{eqn:cover_polymatroid2c}, 
respectively. 

\subsubsection{Rate-Splitting and Partial Interference Decoding for 2-Receiver BC}
The inner bound of Theorem \ref{thm:superposition_DMBC_IB} 
 is implicitly described in  terms of split-rates as well as the desired rates. 
 In principle, the split-rates can be projected away with Fourier Motzkin, but in
practice, this is only possible for very small settings. 
For example,  in the two-user BC with message index set 
$\spfont{E} = \{1,2,12\}$, with the order relation taken to be that of subset inclusion, for each set of input, auxiliary, and coded time-sharing random variables 
$X' = (X,U_1,U_2,U_{12},Q)\in\mathcal{A}_{(\subseteq)}^{\spfont{E}}$, the polyhedron
$ 
\mathcal{P}_{\downarrow}^{(1)}(X';\spfont{E})\cap 
	 		  \mathcal{P}_{\downarrow}^{(2)}(X';\spfont{E})
$
consists of those non-negative rates which satisfy 
\begin{IEEEeqnarray}{rl}\IEEEyesnumber\label{eqn:twouserFM}
	R_1 + R_{12} &\leq I(U_1,U_{12};Y_1|Q) \IEEEyessubnumber\\
	R_2 + R_{12} &\leq I(U_2,U_{12};Y_2|Q) \IEEEyessubnumber \\
	R_1 + R_2 + R_{12} &\leq I(U_2,U_{12};Y_2|Q)  + I(U_1;Y_1|U_{12},Q)\IEEEyessubnumber \\
	R_1 + R_2 + R_{12} &\leq  I(U_2;Y_2|U_{12},Q) + I(U_1,U_{12};Y_1|Q)\IEEEeqnarraynumspace \IEEEyessubnumber  
\end{IEEEeqnarray}
Note that in this two-user setting, the message index set $\spfont{E}$ is as large
as it can be, so necessarily $\spfont{F}=\spfont{E}$.

\subsection{Three-receiver BCs}
\subsubsection{A class with with degraded message sets}
For the three-receiver BC  with message index
set $\spfont{E} = \{1,123\}$ and the degraded receiver pair
  $X\markov Y_1\markov Y_2$, Nair and El Gamal \cite{Nair:2009is}
  determine that the capacity region is given by
   the non-negative rates for which
 \begin{IEEEeqnarray}{rl}	\IEEEyesnumber \label{eqn:NEG_rateregion}
 	R_{123} &\leq\min\{ I(U;Y_2),I(V;Y_3)\}	\IEEEyessubnumber	\\
	R_1 & \leq I(X;Y_1|U)	\IEEEyessubnumber \\
	R_1 + R_{123} &\leq I(V;Y_3) + I(X;Y_1|V) \IEEEyessubnumber
 \end{IEEEeqnarray} 
 for some pair $U,V$  for which $U\markov V \markov X$.
 Further, by optimizing over the choice of auxiliary random variables,
 it can be shown that every rate in the capacity region
 satisfies \eqref{eqn:NEG_rateregion} and
 the additional inequality
  \begin{IEEEeqnarray}{rl}	
	R_1  & \leq I(X;Y_1|V) + I(V;Y_3|U)  \label{eqn:NEG_rateregion2}
 \end{IEEEeqnarray} 
  for some pair $U,V$ for which $U\markov V \markov X$ \cite{Nair:2009is}.
  
The inner bound  of Corollary \ref{corr:fullsuperposition_DMBC_IB} is 
equal to this second capacity characterization.
To achieve this rate region, it suffices to omit coded time-sharing
and consider rate-splitting
 in the context of the enlarged message index set $\spfont{F}= \{1,13,123\}$.
With these restrictions, Theorem \ref{thm:superposition_DMBC_IB}
states that  the rates $(R_1,R_{123})$ are achievable
 if there exist non-negative split-rates
$r_{1\to 1} = \hat{R}_1$, $r_{1\to 13} = \hat{R}_{13}$,
and $r_{123\to 123} = \hat{R}_{123}$ 
so that $R_1 = \hat{R}_1 + \hat{R}_{13}$, $R_{123} = \hat{R}_{123}$,
and
\begin{IEEEeqnarray}{rl}\IEEEyesnumber\label{eqn:nairEG_polyhedron}
	\hat{R}_1 &\leq I(U_1;Y_1|U_{13},U_{123}) \IEEEyessubnumber \\
	\hat{R}_1 + \hat{R}_{13} &\leq I(U_1,U_{13};Y_1|U_{123}) \IEEEyessubnumber \\
	\hat{R}_1 + \hat{R}_{13}  +\hat{R}_{123} &\leq I(U_1,U_{13},U_{123};Y_1) \IEEEyessubnumber \\[5pt]
	\hat{R}_{123} & \leq I(U_{123};Y_2) \IEEEyessubnumber \\[5pt]
	\hat{R}_{13}&\leq I(U_{13};Y_3|U_{123}) \IEEEyessubnumber\label{eqn:uniquedec_constraint}\\ 
	\hat{R}_{13} + \hat{R}_{123} & \leq I(U_{13},U_{123};Y_3)\  \IEEEyessubnumber 
\end{IEEEeqnarray}
for some triple $(U_1,U_{13},U_{123})$
where $X = U_1$ and $U_{123}\markov U_{13} \markov U_1$
form a Markov chain.
With Fourier-Motzkin, and the degraded assumption on 
the first two receivers, this demonstrates the achievability
of any rate pair that satisfies \eqref{eqn:NEG_rateregion} and 
\eqref{eqn:NEG_rateregion2}, with $U=U_{123}$ and $V=U_{13}$.


\subsection{Combination Networks}
Consider a noiseless $K$-user broadcast network
where the input is comprised of the $2^K-1$ components
$(V_S: \emptyset \subset S\subseteq \mL{K})$, indexed by
the non-empty subsets of $\mL{K}$, and 
the $j$th output has noiseless access to those
 components that are indexed by $\spfont{W}_j$
 so that  $Y_j = (V_S : S\in \spfont{W}_j)$.
For each $S\subseteq \mL{K}$, the
component $V_S$ is assumed to be within a finite
alphabet $\mathcal{V}_S$; let $C_S = \log_2|\mathcal{V}_S|$. 

When $K=3$, and the message set demand contains all possible
messages, so that $\spfont{E} = 2^{\mL{K}}$,
the capacity of this channel
was determined in \cite{Grokop:2008ti}.
Corollary \ref{corr:fullsuperposition_DMBC_IB} provides
an alternative proof of achievability of this result
for a particular choice of auxiliary random variables,
as discovered in \cite{romero2016superposition}.
That choice assigns 
\begin{equation}
	U_S\sim\mathrm{Uniform}(\mathcal{V}_S)
	\qquad
	V_S = U_S		\label{eqn:specificAuxRVchoice}
\end{equation}
so that $I(U_\spfont{B};Y_j|U_{\spfont{W}_j\backslash \spfont{B}})
= \sum_{S\in \spfont{B}}C_S$.
With computer-aided Fourier-Motzkin \cite{gattegno2016fourier}, 
the split-rates can be projected away,
yielding the 15 defining inequalities of the capacity region.

When the number of users is arbitrary, but
a symmetry assumption of $C_S = C_{|S|}$ and
$R_S = R_{|S|}$ is imposed across all $S\subseteq \mL{K}$,
then the capacity region was determined in
\cite{Tian:2011kh} and \cite{Salimi:2015eh}.
In this setting, the specific auxiliary variable choice 
\eqref{eqn:specificAuxRVchoice} within 
Corollary \ref{corr:fullsuperposition_DMBC_IB} yields
an alternative proof of the achievability of the capacity 
region, as detailed in \cite{romero2016superposition}.



\section{Binning}
A last element is to consider binning, which allows consideration of 
\emph{arbitrary} input distributions. The central idea is create an excessively
large codebook, with rates $\tilde{R}_S\geq R_S$ 
where each message has a list of codewords of exponential size $2^{n(\tilde{R}_S-R_S)}$,
 rather than a single codeword. If the rate excesses $\tilde{R}_S-R_S$ 
 are sufficiently large, then every message can jointly select a set of codewords
 that appear as though they were jointly
generated with respect to an arbitrary joint distribution, rather than
according to its recursive marginal distributions.

In particular, the excess rates will be $\tilde{R}_S$, with the excess
over the desired rate being $r_S= \tilde{R}_S-R_S$ for each $S\in \spfont{E}$.
The key result is a recursive generalization of the mutual covering lemma of El Gamal and van der Meulen in \cite{el1981proof} (see also
\cite[Lemma 8.1]{Gamal:2012}).

\begin{lemma}[Recursive Mutual Covering Lemma]\label{recursive_mutualcoveringlemma} Let $(U_S:S\in \spfont{E})$
have arbitrary joint distribution $p(u_S:S\in \spfont{E})$. Pick an order on $\spfont{E}$.
With respect to this order, recursively generate length-$n$ vectors 
\[	u_S^n(m_S) \sim \prod_{t=1}^n p(u_{St}|u_{Rt}:R\in \upset{\{S\}}\backslash\{S\}) \]
 for each $m_S\in\mL{2^{nr_S}}$ and each $S\in \spfont{E}$.
Then the probability that $ (u_S(m_S):S\in \spfont{E})	$
is jointly $\epsilon$-typical for some $(m_S:S\in \spfont{E})$ tends to one
as $n\to \infty$ if the non-negative rates $(r_S:S\in \spfont{E})$ satisfy
\begin{align}
		r(\spfont{G}) \geq &\sum_{S\in \spfont{G}} H(U_S|U_{\upset{\{S\}}\backslash\{S\}}) 
				- H(U_{\spfont{G}}) \triangleq \gamma(\spfont{G}),	\label{eqn:mutual_covering_conditions_LCI}
\end{align}
for all up-sets $\spfont{G}\subseteq \spfont{E}$.
\end{lemma}
\begin{proof}See Appendix \ref{sec:recursivecoveringlemmaproof}. \end{proof}
This unbounded polyhedron is a contra-polymatroid,  defined only
over the up-set lattice $\mathcal{F}_{\uparrow}$. This follows as
\begin{itemize}
\item $\gamma(\spfont{G})$ is supermodular: $\gamma(\spfont{F}\cap \spfont{G})  + \gamma(\spfont{F} \cup \spfont{G}) \geq \gamma(\spfont{F}) + \gamma(\spfont{G})$,
a consequence of the submodularity of entropy.
\item  $\gamma(\spfont{G})$ is non-increasing, a consequence of the fact
that conditioning reduces entropy: for $\spfont{F}\subseteq \spfont{G}$,
\begin{align*}
	\gamma(\spfont{G}) {-} \gamma(\spfont{F}) &{=}\left( \sum_{S\in \spfont{G}\backslash \spfont{F}}H(U_S|U_{\upset{\{S\}}\backslash\{S\}})\right) {-} H(U_{\spfont{G}\backslash \spfont{F}}|U_\spfont{F}) \\
			& {\geq} \sum_{S\in \spfont{G}\backslash \spfont{F}}\bigg(H(U_S|U_{\upset{\{S\}}\backslash\{S\}}) {-} H(U_S|U_{\spfont{F}}) \bigg)\\
			&\geq \sum_{S\in \spfont{G}\backslash \spfont{F}}\bigg(H(U_S|U_{\spfont{F}}) {-} H(U_S|U_F) \bigg)\\
			&= 0
\end{align*}
\item $\gamma(\spfont{G})$ is normalized: $\gamma(\emptyset) = 0$ as the sum is vacuous.
\end{itemize}

This new recursive mutual covering lemma provides a basis
for a significant generalization of the Marton achievable scheme for the broadcast channel with private messages (given in \cite[Section 8.6]{el2011network}) to a scheme that incorporates Marton-type coding with general forms of rate-splitting and superposition coding and is applicable to arbitrary groupcasting message sets $\spfont{E}$. Indeed, such a scheme was given in special cases (in number of users and/or message sets), all of which require only the mutual covering lemma, namely, in \cite[Theorem 5]{liang2007rate} for the two-user BC with two private and one common message, in \cite{Nair:2009is} (see \cite[Proposition 8.2]{el2011network}) for the three-user broadcast channel with degraded messages, and in \cite{salman2020diamond} for the $K$-user broadcast channel with the so-called diamond groupcasting message set. While these three rate regions were given as unions of explicit polytopes in message rate space (i.e., after Fourier Motzkin elimination) we provide here an implicit description of the inner bound in terms of the split-rates. Explicit descriptions of the inner bound of the following theorem may be possible in other special cases of the general groupcasting message set $\spfont{E}$ (and choices of $\spfont{F}$).

Note that even the private message case where $\spfont{E} = \{1,2, \cdots , K\}$ and any choice of $\spfont{F} \supset \spfont{E}$ would as such constitute a generalization of Marton's inner bound (which corresponds to the no rate-splitting choice $\spfont{F}=\spfont{E}$)), with that no rate-splitting bound thought to be the best known inner bound for the DM broadcast channel with private messages (see \cite[Section 8.6]{el2011network}).  

\begin{theorem}[Generalization of Marton Coding]\label{thm:martonextension}

For the $K$-receiver DM BC with general message sets, 
the non-negative rates $(R_S:S\in \spfont{E})$ are achievable if,
for a message index superset $\spfont{F}$ with $\spfont{E}\subseteq \spfont{F}$ and with $\spfont{F}$
equipped with a superposition order $\leq$, there
 exist non-negative split-rates $(r_{S\to S'}: S\in\spfont{E},S'\in \spfont{F}, S\subseteq S')$,
where the desired rates satisfy
\begin{equation} R_S = \sum\nolimits_{S\subseteq S': S'\in \spfont{F}} r_{S\to S'}	\quad \text{ for all }S\in \spfont{E}
		\label{eqn:targetrates_BC-mc}\end{equation}
while the reconstructed rates 
\begin{equation} 	
	\hat{R}_{S'} = \sum\nolimits_{S\subseteq S': S\in \spfont{E}}r_{S\to S'}	\quad \text{ for all }S'\in \spfont{F} 
	\label{eqn:reconstructedrates_BC-mc}\end{equation}
	satisfy the binning constraints
\begin{align}	\sum\nolimits_{S\in \spfont{G}} (\tilde{R}_S-\hat{R}_S)	&\geq \gamma (\spfont{G}) \label{eq:martonregionextension} 
\end{align}
for every up-set $\spfont{G}\in\mathcal{F}_{\uparrow}(\spfont{F};\leq)$ with the rates $(\tilde{R}_S : S \in \spfont{F})$ constrained to be within the rate region (with $\bar{X} \triangleq (X,U_{\spfont{F}},Q={\rm const.})$)
\begin{IEEEeqnarray}{c}
			\left(\mathcal{P}_{\downarrow}^{(1)}(\tilde{X};\spfont{F})\cap
		 \cdots \cap \mathcal{P}_{\downarrow}^{(K)}(\tilde{X};\spfont{F})\right)  \label{eqn:packing_constraint_BC-mc} \\
		\IEEEnonumber
\end{IEEEeqnarray}
for some arbitrary distribution $p(u_\spfont{F})$ and $X$ a deterministic function of $U_\spfont{F}$.
\end{theorem}
\begin{proof}
An outline of the proof is given. The proof involves rate-splitting and message reconstruction as in the achievable scheme used to prove Theorem \ref{thm:superposition_DMBC_IB} with a similar codebook generation scheme using superposition coding according to the general superposition order $\leq $ but each codebook is generated to have excess codewords (with the $U_S$ codebook having $2^{n\tilde{R}_S}$ codewords) which are assigned to $2^{n\hat{R}_S}$ equi-sized bins (of size $2^{n(\tilde{R}_S-\hat{R}_S)}$ codewords) indexing the reconstructed messages $\hat{m}_S$. The selection of the set of codewords from the bins that are jointly typical incurs an error whenever such a set of jointly typical codewords does not exist and this encoding error can be made vanishingly small provided that the binning constraints \eqref{eq:martonregionextension} hold as dictated by the recursive mutual covering Lemma \ref{recursive_mutualcoveringlemma}. 
\end{proof}

Prior works that combine superposition coding, rate-splitting with binning include the two-user broadcast channel with two private and one common messages \cite{liang2007rate} (see \cite[Theorem 8.4]{el2011network}), the three-user case with three degraded messages in \cite{Nair:2009is} (but with one receiver employing indirect joint decoding) and, more generally, to the diamond message set groupcasting in the $K$-user broadcast channel in \cite{salman2020diamond}. However, in all of these schemes, the classical mutual covering lemma proved by El Gamal and van der Meulen \cite{el1981proof} (cf. \cite{el2011network}) suffices for their analyses. Moreover, because in these problems Fourier Motzkin elimination was possible, direct descriptions of the rate regions are given therein as unions of polytopes in the space of the rates of the messages.

\section{Conclusion}
With the tools of order theory, we provide
a general inner bound for the DM BC with
an arbitrary number of users and an arbitrary
set of message demands, a  setting
which encompasses any type of
 concurrent groupcasting.

Prior results on the capacity region 
of the two-user broadcast channel with degraded messages
 and a three-user multi-level broadcast channel 
 with two degraded messages can be seen through
  the lens of the general inner bound proposed in this work, 
  as can other superposition coding based achievable 
  schemes with joint desired message and partial interference decoding
  proposed in multiple unicast or private message settings. 

\begin{appendices}

\section{Achievability of reconstructed rates in Theorem \ref{thm:superposition_DMBC_IB}}
The general notion of superposition coding in \cite{romero2017unifying} allows for a unified 
treatment of the achievability of
the reconstructed rate regions in Theorem \ref{thm:superposition_DMBC_IB}.
That is, with the order-theoretic tools described therein, we will 
demonstrate that the reconstructed rates $(\hat{R}_S:S\in\spfont{F})$
are achievable if they are within  
\begin{IEEEeqnarray}{rl}
	&\bigcup\nolimits_{X'\in\mathcal{A}_{(\leq)}^{\spfont{F}}}	 \left(\mathcal{P}_{\downarrow}^{(1)}(X';\spfont{F})\cap 
	 		\cdots \cap \mathcal{P}_{\downarrow}^{(K)}(X';\spfont{F})\right) 
				\label{eqn:reconstructedrateregions_superposition} 
\end{IEEEeqnarray}
where $\mathcal{A}_{(\leq)}^{\spfont{F}}$ is defined in
Definition \ref{def:superposition_admissibleRVs}, and  the 
 polyhedra $\mathcal{P}_{\downarrow}^{(j)}(X';\spfont{F})$, for $j\in\mL{K}$,
  were defined in \eqref{eqn:downsetpolymatroid}.
	
Superposition coding is a random coding strategy
in which the dependencies between the codewords
between distinct message sources follow 
an order on the message index set.
To prove the achievability of \eqref{eqn:reconstructedrateregions_superposition}, 
equip the set $\spfont{F}$ with a superposition order.
Fix a set of input, auxiliary, and coded
time-sharing random variables $(X,U_{\spfont{F}},Q)$
such that $X$ is a deterministic function of 
the auxiliary and coded time-sharing random variables
$(U_{\spfont{F}},Q)$, whose joint distribution factors as
$p(Q)\prod_{S\in\spfont{F}}p(U_S|(U_{S'}:S'\in\spfont{F},S'>S),Q)$.
Let $\{S_1,\ldots,S_M\}$ be an non-increasing
enumeration of $\spfont{F}$. 
First, generate the coded time-sharing sequence
$q^n\sim \prod_{t=1}^n p(q_t)$. Then, for
each $i\in\mL{M}$, and each collection of
messages $m_{\upset{\{S_i\}}} = (m_{S'}: S'\in\spfont{F}, S'\geq S_i)$,
generate the codewords through
$ u_{S_i}^n(m_{\upset{\{S_i\}}} )\sim \prod_{t=1}^np(u_{S,t}|(u_{S',t}:S'\in\spfont{F}, S' > S),q_t)$.
This process can be carried out iteratively
from $i=1$ to $i=M$ as $\upset{\{S_i\}} \subseteq \{S_1,\ldots,S_i\}$
for each $i\in\mL{M}$.

The $j$th receiver jointly decodes the messages
 $m_{\spfont{W}_j^{\spfont{F}}}$, 
 while treating all other messages as noise. To describe
 and analyze this joint decoding rule and its error probabilities,
  let $\delta>0$ be an infinitesimal typicality parameter 
 and $\epsilon>0$ be a target error probability guarantee. With 
 $m\equiv m_{\spfont{W}_j^{\spfont{F}}} $, let  $T(m)$ be the event that
  $(q^n,u_S^n(m_{\upset{\{S\}}}):S\in\spfont{W}_j^{\spfont{F}}, y_j^n)$
is $\delta$-jointly typical with respect to the join distribution 
on $(Q,U_{\spfont{W}_j^{\spfont{F}}},Y)$. The $j$th receiver declares
the message estimates  $\hat{m}$
to be the sent messages if and only if it is the unique set of messages
for which $T(\hat{m})$ occurs. Without loss of generality,
assume that the set message was $m = \M{1}$, 
which is to say that $(m_S =1:S\in \spfont{F})$. An error occurs if either 
a) the event $T(\M{1})$ does not occur, or b) the event
$T(\hat{m})$ occurs from some $\hat{m}\neq \M{1}$. The law
large numbers assures that the probability of the first event 
 vanishes in the block length $n$. For the second
 error event, we show that the probability of 
 certain categories of error events must vanish if
 a certain partial sum-rate condition holds. With the union bound,
 we extend this argument to show that if all such possible
 partial sum-rate bounds hold, then the probability if
 any error must vanish.
 
Let $\spfont{B}$ be any subset of $\spfont{W}_j^{\spfont{F}}$. 
By the joint typicality lemma \cite{Gamal:2012}
one can show that the probability that $T(\hat{m})$
occurs for a wrong message of the form  $\hat{m}_S\neq 1$ for $S\in\spfont{B}$
 and $\hat{m}_S= 1$ for $S\not\in \spfont{B} $ is bounded by 
$2^{-n(I(U_{\spfont{C}};Y|
U_{\spfont{W}_j^{\spfont{F}}\backslash \spfont{C}}) - \epsilon_\delta)}$
for a parameter $\epsilon_\delta$ that tends to zero as $\delta$ vanishes.
Here, $\spfont{C}$ is the smallest down-set within $\spfont{W}_j^{\spfont{F}}$
containing $\spfont{B}$, succinctly denoted by $\spfont{C}= \downset{\spfont{B}}$
for the remainder of this proof \cite{romero2017unifying}.
Together with the union bound, a finite summation yields that the
probability of error vanishes if
$\sum_{S\in\spfont{B}}R_S\leq 
I(U_{\downset{\spfont{B}}};Y|U_{\spfont{W}_j^{\spfont{F}}\backslash \downset{\spfont{B}}})$
for every $\spfont{B} \subseteq \spfont{W}_j^{\spfont{F}}$. 
By removing redundant inequalities, these conditions are 
equivalent to the requirement that
$\sum_{S\in\spfont{B}}R_S\leq 
I(U_{\spfont{B}};Y|U_{\spfont{W}_j^{\spfont{F}}\backslash \spfont{B}})$
for every down-set $\spfont{B}$ in $\spfont{W}_j^{\spfont{F}}$,
with respect to the particular order chosen on $\spfont{F}$.
These are the defining inequalities within \eqref{eqn:downsetpolymatroid}.
Replicating this argument over all superposition-admitting
random variables yields that the rate region 
\eqref{eqn:reconstructedrateregions_superposition} is achievable.

\section{Recursive Mutual Covering Lemma Proof\label{sec:recursivecoveringlemmaproof}}
Let $\mathcal{T}_\epsilon^{(n)}(U_{\spfont{E}})$ be the set of jointly $\epsilon$-typical length-$n$
sequences $(u_S^n:S\in \spfont{E})$ with respect to the joint distribution of $U_{\spfont{E}}$.
Let 
\[	\mathcal{A}=  \{(m_S:S\in \spfont{E}) : (u_S^n(m_S):S\in \spfont{E}) \in \mathcal{T}_\epsilon^{(n)}(U_{\spfont{E}}) \}	\]
be the set of all independently randomly generated vectors which appear
as though they were jointly generated (by being jointly typical).
Chebyshev's inequality supplies
 $P(|\mathcal{A}|= 0)\leq \mathrm{Var}(|\mathcal{A}|)/E(|\mathcal{A}|)^2 $.
The probability mass function governing the distribution
 of a codeword tuple  $(u_S^n(m_s):S\in \spfont{E})$ is independent
  of the message tuple $m\equiv(m_S:S\in \spfont{E})$. Thus 
  each codeword tuple has the same probability of being jointly typical,
  which we define to be
  \[	P(U_{\spfont{E}}^n(m)\in \mathcal{T}_\epsilon(U_{\spfont{E}})) = P(U_{\spfont{E}}^n(\mathbf{1})\in \mathcal{T}_\epsilon(U_{\spfont{E}}))\triangleq p.		\]
By linearity of expectation, $E[|\mathcal{A}|] = 2^{nr(\spfont{E})}p $, where $r(\spfont{E}) \triangleq \sum_{S \in \spfont{E}}r_S$. With $\llbracket \cdot \rrbracket $ denoting the indicator function, introduce
\[		B(m(\spfont{E}),m'(\spfont{E}))	=	\llbracket 	U_{\spfont{E}}^n(m_{\spfont{E}})\in \mathcal{T}_\epsilon(U_{\spfont{E}})	\rrbracket 
		\llbracket 	U_{\spfont{E}}^{n}(m'_{\spfont{E}})\in \mathcal{T}_\epsilon(U_{\spfont{E}})	\rrbracket	\enspace . \]
		Then we may write
\begin{align*}
	E[|\mathcal{A}|^2]
	&=  \sum_{m(\spfont{E}),m'(\spfont{E})}E [ B(m(\spfont{E}),m'(\spfont{E})) ] \\
	&= \sum_{\spfont{D}\subseteq \spfont{E}}\sum_{\substack{m_S \neq m'_S \forall S\in D
	\\ m({\spfont{D}}^c) = m({\spfont{D}}^c)}}E [ B(m(\spfont{E}),m'(\spfont{E})) ]	\\
	&\leq \sum_{\spfont{D}\subseteq \spfont{E}}  2^{n(2r(\spfont{D}) + r(\spfont{E}\backslash \spfont{D}))}p_{\spfont{D}},
\end{align*}
where $p_{\spfont{D}} = E [ B(m(\spfont{E}),m'(\spfont{E})) ]$ with $(m_S=1:S\in \spfont{E})$
and $(m_S' = 1: S\not\in \spfont{D})$ but $(m_S' =2 :S\in \spfont{D})$, more succinctly
stated as $m = \mathbf{1}$ and $m' = \mathbf{1}+\M{1}(\spfont{D})$
where $\M{1}(\spfont{D}) = (\llbracket S\in \spfont{D} \rrbracket : S\in \spfont{E})$ is the indicator
vector for the subset $\spfont{D}$ of $\spfont{E}$.  Note that the terms corresponding
to $\spfont{D}=\spfont{E}$ or $\spfont{E}=\emptyset$  are easily computed: $p_{\spfont{E}} = p^2$ and $p_\emptyset = p$.
In fact, the term corresponding to $\spfont{D}=\spfont{E}$ is simply $[E[|\mathcal{A}|]]^2$, and thus
\[	\mathrm{Var}(|\mathcal{A}|) \leq \sum_{\spfont{D}\subset \spfont{E}}  2^{n(2r(\spfont{D}) + r(\spfont{E}\backslash \spfont{D}))}p_{\spfont{D}}, 	\]
where the inequality in the summation indices is strict.
By the recursive generating procedure,
if $m_S\neq m_{S'}$ for some $S\in \spfont{E}$, then all the vectors
corresponding to $R\in \downset{\{S\}}$ appear as though they were
independently generated, even if the indices $m_R$ and $m_R'$ match.
In other words,
\begin{align}
	p\bigg(U^n(\M{1})=u^n,U^{n}(\M{1}+\M{1}(\spfont{D})) = v^n\bigg)
	&= \nonumber \\ \qquad \prod_{t=1}^n  \prod_{S\in \downset{\spfont{D}}}p(u_{St}|u_{\upset{\{S\}}\backslash\{S\}t})  p(v_{St} & |v_{\upset{\{S\}}\backslash\{S\}t})\notag 
		  \times \nonumber \\   \prod_{S\in \spfont{E}\backslash \downset{\spfont{D}}}p(u_{St} & |u_{\upset{\{S\}}\backslash\{S\}t})
		 \end{align}
		 
for all potential sequences $u^n,v^n$ with  $u_S= v_S$ when $S\not\in \downset{\spfont{D}}$. If 
these potential sequences are in addition also jointly $\epsilon$-typical
(denote the set of such sequences with $\mathcal{T}(\epsilon,\spfont{D})$),
then
\[  p\left(U^n(\M{1})=u^n,U^{n}(\M{1}+\M{1}(\spfont{D})) = v^n\right)
	\leq 2^{-n(d_{\spfont{D}} -\delta_\epsilon/6)}	\]
where
\[	d_{\spfont{D}} = 2\sum_{S\in \downset{\spfont{D}}}H(U_S|U_{\upset{\{S\}}\backslash\{S\}}) +
		 \sum_{S\in \spfont{E}\backslash \downset{\spfont{D}}}H(U_S|U_{\upset{\{S\}}\backslash\{S\}}).	\]
Consider the set of sequences $v^n$ that are jointly $\epsilon-$typical with respect
to the \emph{joint} distribution on $U_{\spfont{E}}$ and have components
in $\spfont{E}\backslash \spfont{D}$ fixed to $v_S=u_S$ for $S\in \spfont{E}\backslash \spfont{D} $, which we denote
by $\mathcal{S}(\epsilon,u_{\spfont{E}\backslash \spfont{D}})$. Then by standard arguments, 
$\log |\mathcal{S}(\epsilon,u_{\spfont{E}\backslash \spfont{D}}) |\leq n H(U_{\spfont{D}}|U_{\spfont{E}\backslash \spfont{D}}).$
Combining the above observations yields a bound on the probability 
that the pair of sequences $U^n(\M{1})=u^n,U^{n}(\M{1}+\M{1}(\spfont{D}))=v^n$ are jointly typical:
\begin{align*}
	p_{\spfont{D}}
	&= \sum_{u_{\spfont{D}}^n,v_{\spfont{D}}^{n}\in \mathcal{T}(\epsilon,\spfont{D})}p\bigg(U^n(\M{1})=u^n,U^{n}(\M{1}+\M{1}(\spfont{D})) = v^n\bigg)\\
	&= \sum_{u_{\spfont{D}}^n,v_{\spfont{D}}^{n}\in \mathcal{T}(\epsilon,\spfont{D})} 2^{-nd_{\spfont{D}}}\\
	&=  \sum_{u_{\spfont{D}}^n: u^n\, \epsilon-\text{typical}} \sum_{v^n\in \mathcal{S}(\epsilon,u_{E\backslash \spfont{D}})}2^{-nd_{\spfont{D}}}\\
	&\leq  2^{nH(U_{\spfont{E}}) + nH(U_{\spfont{D}}|U_{\spfont{E}\backslash {\spfont{D}}})- n d_{\spfont{D}} + \delta_\epsilon/2},
\end{align*}
This bound, combined with  the bound
\[	\frac{1}{n}\log p\geq H(U_{\spfont{E}})- \sum_{S\in \spfont{E}} H(U_S|U_{\upset{\{S\}}\backslash\{S\}}) + \delta_\epsilon/4	\]
yields
\begin{align*}
	\frac{\mathrm{Var}(|\mathcal{A}|)}{E[|\mathcal{A}|]^2 }
		&\leq \sum_{\spfont{D}\subset \spfont{E}} 2^{-n\triangle(D)} 
\end{align*}
where
\[	\triangle(D) = r(\spfont{E}\backslash  \spfont{D}) - \left(\sum_{S\in \spfont{E}\backslash \downset{D}} 
					H(U_S|U_{\upset{\{S\}}\backslash\{S\}}) - H(U_{\spfont{E}\backslash \spfont{D}})\right) - \delta_\epsilon.	\]
Thus the error tends to zero if the conditions 
\begin{align*}
		r(\spfont{G}) \geq &\sum_{S\in \mathcal{X}_{\uparrow}({\spfont{G}})} H(U_S|U_{\upset{\{S\}}\backslash\{S\}}) - H(U_{\spfont{G}}),
\end{align*}
where $ \mathcal{X}_{\uparrow}(\spfont{G})$ is the largest up-set contained within $\spfont{G}$,
hold for all subsets $\spfont{G}\subseteq \spfont{E}$. But not all of these inequalities are necessary.
Observe that as the rates are non-negative, and $ \mathcal{X}_{\uparrow}(\spfont{G})$
is a subset of $\spfont{G}$,
\begin{align}
	r(\spfont{G})  \geq r(\mathcal{X}_{\uparrow}(\spfont{G})) 
		& \geq & \sum_{S\in \mathcal{X}_{\uparrow}({\spfont{G}})} H(U_S|U_{\upset{\{S\}}\backslash\{S\}}) - H(U_{\mathcal{X}_{\uparrow}(\spfont{G})}) \nonumber  \\
		& \geq & \sum_{S\in \mathcal{X}_{\uparrow}({\spfont{G}})} H(U_S|U_{\upset{\{S\}}\backslash\{S\}}) - H(U_{\spfont{G}}).		\end{align}
Thus, the inequality corresponding to $\mathcal{X}_{\uparrow}(\spfont{G})$ implies
the inequality corresponding to $\spfont{G}$. As such, it suffices to only enforce
those inequalities corresponding to the up-sets of $\spfont{E}$.

\end{appendices}

\bibliographystyle{plain}
\bibliography{DoF_BCCM}

 \end{document}